\documentclass[conference,letterpaper]{IEEEtran}
\IEEEoverridecommandlockouts 

\usepackage{amssymb, cite}
\usepackage{bm,bbm}
\usepackage{multirow,color,amsfonts}
\usepackage{tabulary}
\usepackage{subfigure}
\usepackage{graphicx}
\usepackage{setspace}
\usepackage{enumerate}
\usepackage{algorithm,algpseudocode}
\usepackage{comment}

\usepackage[utf8]{inputenc} 
\usepackage[T1]{fontenc}
\usepackage{url}
\usepackage{ifthen}
\usepackage{cite}

\usepackage{mathtools} 

\DeclarePairedDelimiterX{\infdivx}[2]{(}{)}{%
  #1\;\delimsize\|\;#2%
}

\DeclarePairedDelimiter{\norm}{\lVert}{\rVert}

\usepackage{amsthm}
\usepackage{amsmath}

\newtheorem{theo}{Theorem}
\newtheorem{defi}{Definition}

\newtheorem{lem}{Lemma}
\newtheorem{ex}{Example}

\newenvironment{brsm}{
  \left[ \begin{smallmatrix} }{%
  \end{smallmatrix} \right]}

\newcommand{\derya}[1]{\textit{\textcolor{blue}{#1}}}

\newcommand{\Kmatch}{K_{\rm M}}
  
\newcommand{\graphweight}{{q}}

\usepackage{enumitem}

\begin{document}
\title{Distributed Computing of Functions of\\ Structured Sources with Helper Side Information}
\author{
  \IEEEauthorblockN{Derya Malak}
  \IEEEauthorblockA{Communication Systems Department,   EURECOM, France,\\
derya.malak@eurecom.fr
}
\thanks{Funded by the European Union (ERC, SENSIBILITÉ, 101077361). Views and opinions expressed are however those of the author only and do not necessarily reflect those of the European Union or the European Research Council. Neither the European Union nor the granting authority can be held responsible for them.}
}

\maketitle

\begin{abstract}

In this work, we consider the problem of distributed computing of functions of structured sources, focusing on the classical setting of two correlated sources and one user that seeks the outcome of the function while benefiting from low-rate side information provided by a helper node. 
Focusing on the case where the sources are jointly distributed according to a very general mixture model, 
we here provide an achievable coding scheme that manages to substantially reduce the communication cost of distributed computing by exploiting the nature of the joint distribution of the sources, the side information, as well as 
the symmetry enjoyed by the desired functions. Our scheme --- which can readily apply in a variety of real-life scenarios including learning, combinatorics, and graph neural network applications --- is here shown to provide substantial reductions in the communication costs, while simultaneously providing computational savings by reducing the exponential complexity of joint decoding techniques to a complexity that is merely linear. 
\end{abstract}


\section{Introduction}
\label{sec:intro}

The past few years have seen a rising need to speed up computationally-intensive tasks, as well as have witnessed an ever increasing necessity for new parallel processing techniques that efficiently distribute computations across groups of servers. This necessary transition to distributed computing though, has also introduced a variety of challenges that involve accuracy \cite{wang2021price}, computing scalability \cite{soleymani2021analog}, and straggler mitigation \cite{li2021flexible}. Key among these challenges comes in the form of the crippling communication bottleneck of distributed computing, brought about by the astronomical communication costs often required for implementing computing in a distributed manner. This bottleneck is central to our work here.

Following the seminal work in \cite{yao1979some} on the communication complexity of distributed computing, several 
works have sought to minimize the associated communication cost, with some of these works including the recent breakthroughs in distributed linearly separable computation 
\cite{khalesi2022multi}, distributed matrix multiplication
~\cite{jia2021capacity}, 
and others that extended the seminal work in \cite{SlepWolf1973} 
to 
function computation, 
including~
\cite{Kor73} and \cite{feizi2014network}, that consider 
function computation over networks, as well as 
\cite{doshi2007source} 
that study functional rate distortion. 
This same communication cost was naturally also affected by the nature of the computed function. With this in mind, 
we studied distributed computation using structured distributions, see e.g., \cite{Malak2022hyperbin}
.


The above motivate us to explore the joint effect of having structure in data and in sources, and to explore how this structure can help us reduce the aforementioned communication bottleneck of distributed computing. Toward this, we here consider a partially distributed computation problem with two sources of jointly distributed data, with one user that seeks the outcome of a function of the sources, and with one low-rate helper that mediates the computation of the function by providing a small amount of side information on the instantaneous matching of the sources. The sources are here modeled by a mixture distribution. From the perspective of Bayesian inference, such models can accurately capture the behavior of data distributed according to a mixing distribution.

Our aim will be to exploit the structure of the data and of the desired functions, in order to reduce communication (as well as computational/decoding) costs. To do that, while capturing the structure of the partially distributed sources, we will seek to decompose the joint source distribution into a convex combination of integral matchings, and we will do so by exploiting the well-known {\emph{Birkhoff-von Neumann Theorem}} in \cite{caron1996nonsquare}, as well as {\emph{Sinkhorn’s Theorem}} in \cite{sinkhorn1964relationship}.  
We will here focus on the case where the nonmatched distributions correspond to low-probability events, and the case where the side information captures this matching behavior. 

We next describe generalizations of the Birkhoff-von Neumann statistical multiplexing approaches that have been successful in signal processing, wireless and networking applications, including switching theory, e.g., the study of the rate region of flows to compute a switch schedule using a graph-theoretic formulation \cite{4285330}, and best-effort switching services \cite{chang2000birkhoff}. 
Fast converging algorithms for delay-sensitive applications with sparse switching configurations by leveraging Frank-Wolfe methods 
were also studied \cite{valls2021birkhoff}. In \cite{dimakis2010gossip}, 
the Birkhoff-von Neumann theorem was used for the analysis of gossip algorithms. There exist applications in multi-sensor data in signal processing \cite{uccar2019partitioning}, matrix analysis \cite{tropp2022acm}, and stochastic matrix optimization for power amplifiers \cite{safa2022low}. 
For instance, in latency-constrained applications, it is possible to precompute an approximated Birkhoff-von Neumann decomposition offline and then select a permutation matrix at random with probability proportional to its coefficient \cite{hoyos2020approximate}.


All the above will allow us to propose an achievable coding scheme that captures the aforementioned structure in data and functions as well as exploits the side information, 
to reduce the overall communication cost. It is worth noting that the proposed scheme is different than the existing zero-error distributed coding schemes 
\cite{witsenhausen1976zero} whose operating rates are often limited by the network topology. The proposed scheme, as we will see, can attain gains of approximately 40\% over distributed schemes that do not exploit the above structure. In addition, the helper-based approach will be shown to dramatically decrease the decoding complexity of joint distributions versus minimum-entropy decoding (which is NP-hard in general) used in Slepian-Wolf source coding \cite{coleman2005towards}. 

The rest of the paper is structured as follows. In Sect.~\ref{sec:model_defns}, we give a primer on 
graph compression for modeling the computation rate region for two sources jointly distributed according to a mixture. 
In Sect.~\ref{sec:problem_statement}, we detail 
the rate region 
with a helper that can extract the 
matching information 
of the sources for (i) perfectly matched, and (ii) maximally matched sources. 
We conclude in Sect.~\ref{sec:examples} with an example to numerically evaluate the complexity of our technique.

{\bf Notation.} For a random variable $X$ with a finite alphabet $\mathcal{X}$, then $P_X$ will denote its probability mass function (PMF). Similarly, for variables $X_1$ and $X_2$, then $P_{X_1,X_2}$ will denote the joint PMF.  
Let the entropy function of a PMF ${\bf p}$ be $h({\bf p})=-\sum\nolimits_i p_i\log p_i$ where the logarithm is in base $2$, $h(p)$ be the binary entropy function with parameter $p$, 
and $H(X)=\mathbb{E}[-\log P_{X}(X)]$ be the Shannon entropy of $X$ drawn from  $P_{X}$. 
We denote by ${\bf X}_1^n=X_{11},X_{12},\dots, X_{1n}\in \mathcal{X}_1^n$ the length $n$ sequence of $X_1$ sampled from an $n$-fold finite alphabet $\mathcal{X}_1^n$. 
We let $[N]=\{1,2,\dots,N\}$, $N\in \mathbb{Z}^+$. 

\section{Model, Definitions, \& Related Existing Results}
\label{sec:model_defns}
We consider the problem of computing functions of two sources of partially distributed data according to a mixture distribution, and one user that aims to recover the function outcome. 
%
%
Our model relies on the availability of the side information that captures the matching information extracted from the mixture, which we will elaborate later on, just below, at the source sites, which we leverage via one helper possibly with a limited rate. 
The user can recover both sources in an error-free manner.
Our main purpose is, for sources and functions with described structures, to characterize an achievable rate region for asymptotically lossless computation.
Given the existence of schemes, e.g., \cite{SlepWolf1973} and \cite{korner1979encode}, 
that aim to recover data or a function 
of data in a similar 
manner and have higher sum rates than our scheme, the two advantages of our model over these pertinent ones are (i) a tighter rate region, which is 
critical in modern applications because of symmetries, and (ii) complexity of decoding, not only alleviating the distributed computing systems from performing calculations beyond what the user seeks to compute 
but also eliminating the prohibitive complexity of minimum-entropy decoding.

In what follows, we describe some of the fundamental ingredients that we will use, first with respect to some basic principles behind the graph entropy approach that will allow us to capture some of the interdependencies between the sources and the computed function, and then we will discuss the mixture model that captures the joint source distribution.

\subsection{Representation of Source Characteristic Graphs} 
\label{sec:CharacteristicGraphs}

In this section, we review the fundamental limits of asymptotically lossless compression for computation, which can be realized using the concept of \emph{characteristic graphs}. 
Source one builds a characteristic graph $G_{X_1}=(V_{X_1},E_{X_1})$ \cite{Kor73} for computing $f(X_1,X_2)$ to distinguish the source outcomes that yield a different output for any value of $X_2$, and similarly for source two. Vertices are the sample values, i.e., $V_{X_1}=\mathcal{X}_1$, and the edges are determined as follows. Given two arbitrary vertices $(x_1^{k_1}, x_1^{k_2}) \in \mathcal{X}_1^2$ in $G_{X_1}$ such that $k_1\neq k_2$, if $\exists$ a vertex $x_{2}^{l}\in \mathcal{X}_{2}$ in $G_{X_{2}}$ such that $P_{X_1,X_2}(x_1^{k_1},x_2^{l})P_{X_1,X_2}(x_1^{k_2},x_2^{l})>0$ and the function satisfies $f(x_1^{k_1},x_{2}^{l})\neq f(x_1^{k_2},x_{2}^{l})$, 
then $(x_1^{k_1},x_1^{k_2})\in E_{X_1}$. Otherwise, $(x_1^{k_1},x_1^{k_2})\notin E_{X_1}$.

We let $c_{G_{X_1}}(X_1)$ be a valid coloring of $G_{X_1}$, where a valid coloring 
is such that any two vertices of $G_{X_1}$ that share an edge are assigned distinct colors. As one would expect, the characteristic graphs built by the sources 
are correlated.

For simultaneous encoding of multiple instances of a source, source $m\in [2]$ similarly builds the $n$-th power of $G_{X_m}$, i.e., $G_{{\bf X}_m}^n$. We note that $G_{{\bf X}_m}^n=(V_{X_m}^n,E_{X_m}^n)$ is an OR graph such that $V_{X_m}^n=\mathcal{X}_m^n$ and if $(x_{mi}^{k_1},x_{mi}^{k_2})\in E_{X_m}$ for some coordinate $i\in [n]$, then $({\bf x}_{m}^{k_1},{\bf x}_{m}^{k_2})\in E_{X_m}^n$. 

The characteristic graph entropy of $X_1$ is given by \cite{Kor73} $$H_{G_{X_1}}(X_1)=\lim_{n\to\infty}\, \min\limits_{c_{G_{{\bf X}_1}^n}}\,\frac{1}{n} H(c_{G_{{\bf X}_1}^n}({\bf X}_1))\ ,$$ 
where the minimization is over the set of all valid colorings $c_{G_{{\bf X}_1}^n}({\bf X}_1)$ of $G_{{\bf X}_1}^n$. 
%
Similarly, conditional graph entropy \cite{OR01} and joint graph entropy \cite{feizi2014network} can be determined.

\subsection{Mixture Distribution in the Discrete Parameter Domain}
\label{sec:mixture}

As we outlined above, the two sources are partially distributed based on a mixture distribution. 
Generally, such mixture distribution results from assuming that a random variable $Y$ is distributed according to some parametrized distribution $P_{Y|\Theta}$ with an unknown parameter $\Theta$  with 
the  latent distribution $P_{\Theta}$ \cite{rover2017discrete}. 
In this context, the unconditional distribution $P_Y$ results from marginalizing over $P_{\Theta}$: 
\begin{align}
\label{mixture_distribution}
P_Y(y)= \sum\limits_{\theta} P_{\Theta}(\theta )\,P_{Y|\Theta}(y|\theta )  = \sum\limits_{l\in[L]} q_l\,P_{Y|\Theta}(y|\theta_l )\ ,
\end{align}
where $\Theta\sim (q_1,\dots,q_L)$ is distributed over $ (\theta_1,\dots,\theta_L)$. 

Our interest in mixture distributions stems from two main properties. The first has to do with the fact that such mixture distributions span a broad range of scenarios that are of particular interest in modern science.  
Mixture distributions arise in Bayesian inference where a hypothesis is updated as more information becomes available. 
The second aspect relates to the fact that albeit broad, this mixture structure is endowed with properties that are easier to analyze and expand. One such property relates to the renowned Birkhoff's algorithm. 


Birkhoff's algorithm is a greedy algorithm that receives as input a bistochastic matrix 
$W=(w_{ij})\in\mathbb{R}_{\geq 0}^{N\times N}$ (meaning that each of whose rows and columns sums to $1$), and returns 
a Birkhoff-von Neumann decomposition of $W$, instantiated by the Birkhoff–von Neumann Theorem \cite{caron1996nonsquare}, where $W$ is 
a sum of permutation matrices with non-negative weights \cite{aziz2020simultaneously}. 
The run-time complexity of Birkhoff's algorithm is $O(N^2)$.  
This algorithm is useful in our particular case with partially distributed sources, in providing a decomposition for the mixture distribution, which we will detail in Sect. \ref{sec:representation_matchings}.

Birkhoff's algorithm has applications in fair random assignment,  alleviates the problem of solving challenging linear systems with unstructured and indefinite coefficient matrices, and
improves efficiency and scalability of parallel computing 
\cite{benzi2000preconditioning}. 
It also helps us in our goal to derive the rate region in Sect. \ref{sec:problem_statement} for distributed computing of 
functions.

\section{Problem Statement and Results}
\label{sec:problem_statement}

We here consider a pair of partially distributed sources $X_1,X_2$ that accept a joint PMF $P_{X_1,X_2}(x_1,x_2)>0$ for all $(x_1,x_2)$. 
Our goal is to derive the rate region for the distributed computation of $f(X_1,X_2)$ 
exploiting 
the fact that $(X_1,X_2)$ is distributed according to a mixture PMF (cf.~\ref{mixture_distribution}). Such mixture assumption implies that the sources are matched with a high probability. Equivalently, we can see that for $\mathcal{B}_{X_1,X_2}=(\mathcal{X}_1,\mathcal{X}_2,E)$ being the bipartite graph representation of $P_{X_1,X_2}$ --- with vertex sets $\mathcal{X}_1$ and $\mathcal{X}_2$ representing the individual source outcomes, and edge set $E$ whose weights capture $P_{X_1,X_2}$ --- then the vertices of $\mathcal{B}_{X_1,X_2}$ are non-matched with a vanishing probability. 
To derive the sought rate region, we next detail how to represent and then exploit this structural decomposition of mixture distributions.

\subsection{Representation of Matchings}
\label{sec:representation_matchings}

To now explore and exploit the implications of having integral matchings, let us consider $X_2$. 
Let $X_2=\pi(X_1,\Theta)$ such that the permutation index $l\in [L]$ leads to a matching configuration described by the mixture parameter $\Theta=\theta_l$, as introduced in (\ref{mixture_distribution}). 
We let ${\bf q}=(q _{1},\ldots ,q _{L})\in\mathbb{R}^L_{>0}$ where $q_l=P_{\Theta}(\theta_l)$ such that $\sum _{l\in[L]}q _{l}=1$, for an integer $L\in O(N^2)$  \cite{caron1996nonsquare}, and a set of $N\times N$ permutation matrices $\{B_l,\, l\in[L]\}$.  
Let $\pi_l: \mathcal{X}_1\to \mathcal{X}_1$ be a one-to-one and onto function that permutes the elements of $\mathcal{X}_1$ to provide a perfect matching where every vertex is adjacent to exactly one edge 
in $\mathcal{B}_{X_1,X_2}$, which simply means that $X_2=\pi_l(X_1)$. 
Given the perfect integral matching structure in $B_l$ given $l\in [L]$ and the value of $X_1$, say $X_1=x_{1}^k$, the source $X_2$ takes the value $\pi_l(x_{1}^k)$ and the function $f(X_1,X_2)$ in takes the form 
\begin{align}
\label{general_function_mixture}
f(x_{1}^k,\pi_l(x_{1}^k)) \ ,\quad \forall k\in [N]\ ,
\end{align}
with probability $q_l \mathbb{P}(X_1=x_{1}^k\, \vert \theta_l)$. 
We note that the outcome of $f(X_1,X_2)$ is completely determined by $B_l$ and $X_1$.

We next state our lemma, which allows a standard form for any joint PMF with positive entries.

\begin{lem}
\label{lemma:weighted_sum}
Provided that $P_{X_1,X_2}(x_1,x_2)>0$ for all $(x_1,x_2)$, the joint PMF $P_{X_1,X_2}$ can be written as a weighted sum of generalized permutation matrices:
\begin{align}
\label{joint_distribution_mixture}
P_{X_1,X_2}(x_1,x_2)=\sum\limits_{l\in[L]}q_{l}\,P_{X_1,X_2\vert \Theta}(x_1,x_2\vert \theta_l)\ .
\end{align}
\end{lem}

\begin{proof}
Under our assumption that $P_{X_1,X_2}(x_1,x_2)>0$ for all $(x_1,x_2)$, then we first employ Sinkhorn’s Theorem \cite{sinkhorn1964relationship} to note that $P_{X_1,X_2}$ can be mapped into a doubly stochastic $W=D_1 P_{X_1,X_2} D_2$ for unique diagonal matrices $D_1=(d_{1,ij})$ and $D_2=(d_{2,ij})$ in $ \mathbb{R}_{>0}^{N\times N}$. 
Then we proceed to employ the famous Birkhoff–von Neumann Theorem \cite{caron1996nonsquare} to note that $W=\sum\nolimits_{l\in[L]}q _{l}B_{l}$, and $P_{X_1,X_2}$ is a weighted sum of generalized permutation matrices, thus taking the 
form in (\ref{joint_distribution_mixture}), 
where $P_{X_1,X_2\vert \Theta}(\cdot,\cdot \vert \theta_l)=D_1^{-1}B_{l}D_2^{-1}$ is a generalized permutation matrix obtained from $\{B_l,\, l\in[L]\}$ where $L\in O(N^2)$, and the matrices $D_1^{-1}$ and $D_2^{-1}$ are such that $D_1^{-1}=(d^{-1}_{1,ij})$ and $D_2^{-1}=(d^{-1}_{2,ij})$ in $ \mathbb{R}_{>0}^{N\times N}$, and the mapping $\pi_l: \mathcal{X}_1\to \mathcal{X}_1$ from $X_1$ to $X_2$ captures $P_{X_1,X_2\vert \Theta}(x_1,x_2\vert \theta_l)$.
\end{proof}


\subsection{Identifying Structures of $\mathcal{B}_{X_1,X_2}$ via Matchings}
\label{sec:identification_structures}

There exist techniques to exploit the combinatorial structure of $\mathcal{B}_{X_1,X_2}$ when it has multiple connected components, see e.g., the G{\'a}cs-K{\"o}rner-Witsenhausen common information (GKW-CI)  \cite{gacs1973common
}. 
However, if $\mathcal{B}_{X_1,X_2}$ has only one bipartition, GKW-CI cannot be extracted, i.e., GKW-CI is zero. 

A helper-based scheme can still provide a low-complexity distributed coding technique even when $\mathcal{B}_{X_1,X_2}$ has one bipartition, despite being sub-optimal in terms of its operating rate.  
We are motivated to contemplate a matching-based helper as maximum matching in 
graphs can be determined using polynomial time algorithms   \cite{edmonds1965paths}, which can be exploited to facilitate the Birkhoff-von Neumann decomposition. 
%
When $P_{X_1,X_2}$ is a mixture distribution satisfying (\ref{joint_distribution_mixture}) that accepts a matching-based decomposition, a helper can be used to distinguish the matched and non-matched vertices of $\mathcal{B}_{X_1,X_2}$. 
%
In the following, we devise a helper-based model to extract the matching information from $\mathcal{B}_{X_1,X_2}$, and denote this helper variable by $\Kmatch$. 
For the described partial distributed setting, Theorem \ref{RateRegion_matching} provides the rate region for computing  (\ref{general_function_mixture}). 

\begin{theo}\label{RateRegion_matching}
({\bf A matching-based computation sum rate.}) 
For computing $f(X_1,X_2)$, where $P_{X_1,X_2}$ is given by the mixture PMF in (\ref{joint_distribution_mixture}), with a helper that extracts the perfect matching between the sources, there exists a low-complexity zero-error encoding and decoding of $X_1$ and $\Theta$ that operates at rates
\begin{align}
R_{H}\geq H(\Kmatch)=H(\Theta),\quad
R_1\geq H_{G_{X_1}}(X_1\vert \Theta)\ ,\nonumber    
\end{align}
yielding the following total rate to compute $f(X_1,X_2)$:
\begin{align}
R_M \geq H(\Kmatch)+H_{G_{X_1}}(X_1\vert \Theta) \ ,
\end{align}
where $H_{G_{X_1}}(X_1\vert \Theta)$ denotes the conditional characteristic graph entropy of $X_1$ for computing $f(X_1,X_2)$ given $\Theta$.  
\end{theo}

\begin{proof}
In the partially distributed setting, to be able to compute $f({\bf X}_1^n,{\bf X}_2^n)$ accurately, the perfect matching variable ${\bm \Theta}^n$ where $\Theta_i\sim {\bf q}$ for all $i\in [n]$ should be made available to the user (or via side information), which requires an asymptotic rate $R_H\geq H(K_M)=H(\Theta)$. Knowledge of ${\bm \Theta}^n$ determines the jointly typical coloring sequence pairs $(c_{G_{X_1}}({\bf X}_1^n),\,c_{G_{X_2}}({\bf X}_2^n))$ to be compressed asymptotically. 

Given that the perfect matching variable $\Theta\sim {\bf q}$ is known at the user, 
such that $X_2=\pi(X_1,\Theta)$, it is sufficient for the user if only one source transmits. Let us assume source one is selected. 
Source one builds the characteristic graph $G_{X_1}$ to compute $f({\bf X}_1^n,{\bf X}_2^n)=f({\bf X}_1^n,\pi({\bf X}_1^n,{\bf\Theta}^n))$ given ${\bf\Theta}^n$. This requires, following the notion of the conditional graph entropy, as detailed in \cite{OR01}, an asymptotic rate of $H_{G_{X_1}}(X_1\vert \Theta)$.

The rate needed from source one to compute $f(X_1,X_2)$ is 
\begin{align}
\label{rate_one_source_permutation_invariant}
R_1
\geq\sum\limits_{l\in[L]}  q_l H(f(X_1,X_2)\,\vert \, X_2=\pi_l(X_1)) \ ,
\end{align} 
enabling the partial distributed computation of $f(X_1,X_2)$ at an asymptotic rate upper bounded by $H(f(X_1,X_2)\, , \,  \Theta)$. 
\end{proof}

The rate region in Theorem \ref{RateRegion_matching} is encompassed by that of optimal distributed functional compression given in \cite{feizi2014network}. 
On the other hand, when $P_{X_1,X_2}$ accepts the decomposition in (\ref{joint_distribution_mixture}), a helper-based model to extract the matching information has advantages over Slepian-Wolf coding. Not only it eliminates the complexity associated with joint typicality decoding \cite{cover2012elements} with exponential complexity, but also provides an almost lossless compression asymptotically, given by Theorem \ref{RateRegion_matching}. 

We next impose additional structure on the distributed sources by exploiting the maximal coupling construction. 

\subsection{Extracting Matched versus Non-Matched Vertices}
\label{sec:matching_nonmatching}

In this section, we capture the matchings between the pair of partially distributed sources $X_1$ and $X_2$, which are maximally coupled. 
A maximal coupling between 
a pair $(X_1,X_2)$ 
maximizes $\mathbb{P}(X_1 = X_2)$ subject to the marginal PMFs $X_1\sim {\bf p}_1
$ and $X_2\sim {\bf p}_2
$. 
%
Let $C({\bf p}_1,\,{\bf p}_2)$ be the set of all joint PMFs of $X_1\sim {\bf p}_1$ and $X_2\sim {\bf p}_2$. 
Elements $M=[m_{ij}]~\in\mathbb{R}_{\geq 0}^{|\mathcal{X}|\times |\mathcal{X}|}$ of $C({\bf p}_1,\,{\bf p}_2)$ are couplings of ${\bf p}_1$ and ${\bf p}_2$:
\begin{align}
C({\bf p}_1,\,{\bf p}_2)\triangleq\Big\{
\big[m_{ij}\,:\sum_{j\in |\mathcal{X}|} m_{ij}=p_{1i},\,\,
\sum_{i\in |\mathcal{X}|} m_{ij}=p_{2j}\big]\Big\} \ .\nonumber
\end{align}

A coupling of $(X_1,\,X_2)$ that maximizes $\mathbb{P}(X_1=X_2)$ 
is called a maximal coupling which is formally stated next.
\begin{lem}\label{maximumcoupling}
Maximal coupling $(X_1, X_2 )$ subject to the marginal distributions $X_1\sim {\bf p}_1$ and $X_2\sim {\bf p}_2$ satisfies 
\begin{align}
\label{MaximalCoupling}
\norm{{\bf p}_1-{\bf p}_2}_{\rm{TV}}= \inf\ [\mathbb{P}(X_1 \neq X_2): \,C({\bf p}_1,\,{\bf p}_2)] \ ,
\end{align}
where the measure $\norm{{\bf p}_1-{\bf p}_2}_{\rm{TV}}$ is the total variation distance between the PMFs 
of $X_1$ and $X_2$.
\end{lem} 

From Lemma \ref{maximumcoupling}, if the sources $X_1$ and $X_2$ are maximally matched, the total variation distance between them is minimum.

Maximal coupling of the pair $(X_1,X_2)$ becomes relevant when the discrepancy between ${\bf p}_1$ and ${\bf p}_2$ is bounded above, e.g., the PMFs coincide 
in the first two decimal points. Intuitively, this coupling 
can be exploited to further reduce the communication cost for distributed computing.

Recall that Theorem \ref{RateRegion_matching} derives a schedule-based lower bound on the sum rate given by Lemma \ref{lemma:weighted_sum}, irrespective of the coupling of $(X_1,X_2)$. 
Next, Theorem \ref{RateRegionPermutationInvariant_maximal_coupling} 
provides a lower bound on the sum rate by assuming an additional structural correlation between $X_1$ and $X_2$ through their maximal coupling.

\begin{figure}[t!]
\centering
\includegraphics[width=0.35\columnwidth]{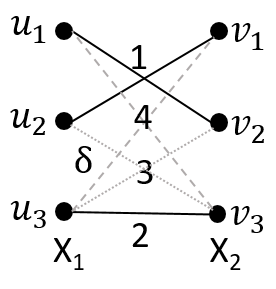}
\caption{Bipartite graph for Example \ref{SymmetricHighError} where the function outcomes are indicated on the edges, the bipartite vertex sets denote the outcomes $\mathcal{X}_1$ and $\mathcal{X}_2$, respectively, and each dotted or dashed edge has a probability $\delta$.}
\label{BipartiteGraphs_HighDelta}
\end{figure}

\begin{theo}\label{RateRegionPermutationInvariant_maximal_coupling}
({\bf A maximally-coupled computation sum rate.}) 
Assume that the sources $X_1$ and $X_2$ are maximally coupled. 
Then, there exists a low-complexity zero-error encoding and decoding of $f(X_1,X_2)$ and $\Theta$ with a helper that extracts the perfect matching between the sources and operates at rates
\begin{align}
R_H+\sum\limits_{m\in[2]}R_m\geq H(\Theta)+\sum\limits_{l\in [L]}q_l \big(h(\delta_l)+\delta_l^c H(T_l)\nonumber
\end{align}
\begin{align}
\label{rate_maximal_coupling}
+\delta_l [H(V_l)+H(W_l)]\big) \ ,
\end{align}
where $V_l$, $W_l$, and $T_l$ are independent integer-valued 
variables with distributions given according to (\ref{Coupling_distributions}) with $\delta$ is substituted with $\delta_l=\mathbb{P}(\pi_l(X_1)\neq X_2)$ for given $l\in [L]$, and $\delta_l^c=1-\delta_l$.
\end{theo}

\begin{proof}
To prove this result, we provide a construction for maximal coupling from 
\cite[Ch. 4.12]{grimmett2001probability}.  
There exists a pair $X_1'$ and $X_2'$ having the same marginals as $X_1$ and $X_2$ such that $\mathbb{P}(X_1'=X_2')= 1-\norm{{\bf p}_1-{\bf p}_2}_{\rm{TV}}=1-\delta$. 
Let $U\sim$Bern$(1-\delta)$ 
and $V$, $W$, $T$ be independent integer-valued variables with respective masses at $k=1,\dots,|\mathcal{X}|$:
\begin{align}
\label{Coupling_distributions}
\mathbb{P}(T=k)&\triangleq \frac{\min[p_{1k},\,p_{2k}]}{(1-\delta)},\\
\mathbb{P}(V=k)\triangleq \frac{[p_{1k}-p_{2k}]^+}{\delta},&\quad
\mathbb{P}(W=k)\triangleq \frac{[p_{2k}-p_{1k}]^+}{\delta} \ .\nonumber
\end{align}
Then, the random variables $X_1'$ and $X_2'$ defined as
\begin{align}
\label{maximally_coupled_discrete_RVs}
X_1' \triangleq U T + (1 - U) V,  \quad
X_2' \triangleq U T + (1 - U) W
\end{align}
have the required marginals, and are maximally coupled such that $\mathbb{P}(X_1'=X_2')=\mathbb{P}(U=1)=1-\delta$. Furthermore, $X'_1$ and $X'_2$ are independent when $X'_1\neq X'_2$ as their sets are disjoint. 
For the proof, we refer the reader to \cite[Ch. 4.12]{grimmett2001probability}. 

%
Let us consider the possible matchings described by $\Theta\sim {\bf q}$, where the matched vertices are invariant up to 
permutations. For each $l\in[L]$ with probability $q_l=P_{\Theta}(\theta_l)$, we have
\begin{align}
\pi(X_1,\theta_l)=\pi_l(X_1)&= U_l T_l + (1 - U_l) V_l \ ,\nonumber\\
X_2&= U_l T_l + (1-U_l) W_l \ , \quad l\in [L] \ ,
\end{align}
where $U_l\sim$Bern$(1-\delta_l)$, and $V_l$, $W_l$, and $T_l$ are independent integer-valued variables with respective masses at $k\in [|\mathcal{X}|]$ according to (\ref{Coupling_distributions}) where $\delta$ being replaced by $\delta_l$.

A maximal coupling of $X_1$ and $X_2$ results in $$\mathbb{P}(X_2=\pi(X_1,\Theta))=\sum\limits_{l\in[L]} q_l \mathbb{P}(U_l=1) =\sum\limits_{l\in[L]} q_l \cdot \delta_l^c\ ,$$
where $\delta_l^c=1-\delta_l$. 
The parts corresponding to $U_l=1$ versus $U_l=0$  denote the matched and the non-matched components, respectively. This yields the sum rate needed from the helper and both sources to compute $f(X_1,X_2)$, as given in (\ref{rate_maximal_coupling}).
\end{proof}

To contrast our model with the state-of-the-art, we next consider an example. 
Due to limited space, we deferred several 
examples to the extended version of the draft \cite{Malak2022Structured}.

\section{Example with Structured Sources}
\label{sec:examples}
We next study an example where a helper leverages the structure of $\mathcal{B}_{X_1,X_2}$, 
and provides the necessary rate to distinguish the matching information, $K_M$, in $\mathcal{B}_{X_1,X_2}$. 
%
%
Extracting this information alleviates the complexity of distributed computing provided that the non-matched distributions correspond to low-probability events. Given $K_M$, one source (or both) needs to send a refinement to identify the function outcome.

In the following, we denote by $F$ the table of function outcomes, where the coordinates match the coordinates of $P_{X_1,X_2}$ with entries ordered in an increasing fashion. 

\begin{ex}\label{SymmetricHighError}Consider the probability matrix and the table of 
function outcomes given as follows:
\begin{align}
P_{X_1,X_2} = \frac{1}{3} \begin{bmatrix} 0 & 1-\delta & \delta \\ 1-\delta & 0 & \delta \\ \delta & \delta & 1-2\delta\end{bmatrix},\quad F= \begin{bmatrix} X & 1 & 4 \\ 1 & X & 3\\ 4 & 3 & 2\end{bmatrix},  \nonumber 
\end{align}
where $\delta\in (0,0.5)$. The bipartite graph $\mathcal{B}_{X_1,X_2}$ is shown in Fig. \ref{BipartiteGraphs_HighDelta} where the vertices are listed in the presented order.

{\bf Fully distributed coding.} The entropy of the function for given $P_{X_1,X_2}$ and $F$ is $H(f(X_1,X_2)) = h\left(\frac{2-2\delta}{3},\frac{1-2\delta}{3},\frac{2\delta}{3},\frac{2\delta}{3}\right)$. A trivial rate upper bound is $H_{G_{X_1}}(X_1)+H_{G_{X_2}}(X_2) = 2\log 3$. The sum rate required to compute the function in the case of no helper is $H_{G_{X_1}}(X_1)+H_{G_{X_2}}(X_2\vert X_1) = \log(3)+\frac{2}{3} h(\delta)+\frac{1}{3} h(\delta,\delta,1-2\delta)$.   

{\bf Partially distributed coding via extracting matchings.} The helper decomposes $\mathcal{B}_{X_1,X_2}$ into a perfect matching and a non-matched graph. More specifically, $P_{X_1,X_2}$ can be described by the following mixture distribution: 
\begin{align}
P_{X_1,X_2} 
=\graphweight_1 \begin{bmatrix} 0 & \frac{1-\delta}{3-4\delta} & 0 \\ \frac{1-\delta}{3-4\delta} & 0 & 0\\ 0 & 0 & \frac{1-2\delta}{3-4\delta}\end{bmatrix}+\graphweight_2 \begin{bmatrix} 0 & 0 & \frac{1}{4} \\ 0 & 0 & \frac{1}{4}\\ \frac{1}{4} & \frac{1}{4} & 0\end{bmatrix},\nonumber
\end{align}
where $\delta\in(0,0.5)$, the first matrix describes a perfect matching, denoted by $\Kmatch=0$ with probability $\graphweight_1=1-\frac{4\delta}{3}$, and the second matrix describes a low-probability event, denoted by $\Kmatch=1$ with  $\graphweight_2=\frac{4\delta}{3}$. Hence, the rate required from the helper to distinguish between these two is $H(\Kmatch)=h\left(\graphweight_1\right)$. 
Given $\Kmatch=0$, only 1 source needs to transmit, which requires a rate of $h\Big(\frac{2-2\delta}{3-4\delta}\Big)$ to determine the function outcome. 
Given $\Kmatch=1$, source 1, $X_1$, needs 2 colors for $u_1$ and $u_2$ to distinguish the outcomes $4$ and $3$ 
each with a probability $\frac{1}{4}$, and $u_3$, which has a probability $\frac{1}{2}$, does not need to be distinguished from $\{u_1,u_2\}$ (no shared edges between $\{u_1,u_2\}$ and $u_3$). Similarly for $X_2$. Hence, given $\Kmatch=1$, each source is required to send at an asymptotic rate of $h\left(\frac{3}{4}\right)$ bits per use. Hence, the sum rate required to compute $F$ with a helper that exploits the matching information of $\mathcal{B}_{X_1,X_2}$ is 
\begin{multline}
R_M\geq H(\Kmatch)+\mathbb{P}(\Kmatch=0)H_{G_{X_1}}(X_1\vert \Kmatch=0) \nonumber\\ 
+\mathbb{P}(\Kmatch=1)(H_{G_{X_1}}(X_1\vert \Kmatch=1)+(H_{G_{X_2}}(X_2\vert \Kmatch=1))\nonumber\\
=h\left(\graphweight_1\right)+\graphweight_1 h\Big(\frac{1-2\delta}{3-4\delta}\Big)+\graphweight_2\Big(h\left(\frac{3}{4}\right)+h\left(\frac{3}{4}\right)\Big)\ .\nonumber
\end{multline}
\end{ex}

From above, our matching-based approach has the best approximation for small $\delta$, where the gain of our model over the fully distributed coding setting that exploits the structure of the function but not the source is $\% 42$ and the loss versus 
the fundamental limit $H(f(X_1,X_2))$ is at most $\% 26$.

\bibliographystyle{IEEEtran}
\bibliography{references}

\end{document}